\documentclass[letterpaper, 10 pt, conference]{ieeeconf}  

\IEEEoverridecommandlockouts                              

\usepackage{graphicx}
\usepackage{amsmath,amssymb,amscd,bbm}
\usepackage{bm}
\usepackage{float}
\usepackage{threeparttable}
\usepackage{color}
\usepackage{hyperref}
\usepackage{geometry}
\newtheorem{theorem}{Theorem}

\newtheorem{lemma}{Lemma}

\newtheorem{remark}{Remark}

\title{\LARGE \bf
Two-stage solution for ancilla-assisted quantum process tomography: error analysis and optimal design
}

\author{Shuixin Xiao$^{1,2,3}$, Yuanlong Wang$^{4}$, Daoyi Dong$^{2,3}$ and Jun Zhang$^{1}$
\thanks{*This research was supported by the National Natural Science Foundation
	of China (62173229, 12288201) and the Australian Research Council's Future Fellowship funding scheme under
	Project FT220100656.}
\thanks{$^{1}$University of Michigan – Shanghai Jiao Tong University Joint Institute, Shanghai Jiao Tong
	University, Shanghai 200240, China.
        {\tt\small xiaoshuixin@sjtu.edu.cn, zhangjun12@sjtu.edu.cn}}%
\thanks{$^{2}$ School of Engineering, Australian National University, ACT 2601, Australia.
        	{\tt\small  daoyi.dong@anu.edu.au}}%
\thanks{$^{3}$School of Engineering and Technology, University of New South Wales, Canberra ACT 2600, Australia.
}%
\thanks{$^{4}$ Key Laboratory of Systems and Control, Academy of Mathematics and Systems Science, Chinese Academy of Sciences, Beijing 100190, China.
	{\tt\small wangyuanlong@amss.ac.cn}}%
}

\begin{document}

\maketitle
\thispagestyle{empty}
\pagestyle{empty}

\begin{abstract}
Quantum process tomography (QPT) is a fundamental task to  characterize the dynamics of quantum systems. In contrast to standard QPT, ancilla-assisted process tomography (AAPT) framework introduces an extra ancilla system such that a single input state is needed.
 In this paper, we extend the  two-stage solution, a method originally designed for standard QPT, to perform 
 AAPT. Our algorithm has  $O(Md_A^2d_B^2)$ computational complexity   where $ M $ is the type number of the measurement operators, $ d_A $ is the dimension of the quantum system of interest, and $d_B$ is the dimension of the ancilla system. Then we establish 
an error upper bound and further discuss the optimal design on the input state in AAPT. A numerical example on a phase  damping process demonstrates the
effectiveness of the optimal design and illustrates the theoretical error analysis.
\end{abstract}

\section{Introduction}
Significant advancements have been made in quantum information technologies, such as quantum computing \cite{qci}, quantum communication \cite{ren2017ground}, quantum sensing \cite{RevModPhys.89.035002,PhysRevA.103.042418,PhysRevApplied.17.014034}, and quantum control \cite{dong2010quantum,dong2022quantum,shu1,Dong2023} over the past few decades. In the practical implementation of these technologies, there exists a fundamental problem of characterizing unknown quantum dynamics, known as quantum process tomography (QPT). QPT is also crucial for other tomography tasks, including quantum state tomography (QST) \cite{Qi2013,Hou2016,MU2020108837,Qi2017,wiseman2009quantum} and quantum measurement device calibration \cite{xiaocdc2021,wang2019twostage,xiao2021optimal,Xiao2023}.

Many references have studied  Hamiltonian tomography \cite{PhysRevLett.108.080502,PhysRevA.95.022335,9026783,zhang1,8022944,unitary}, a reduced version of QPT in  a closed quantum system.
For an open quantum system, Refs. \cite{cptp2018} and \cite{surawy2021projected} proposed iterative projection algorithms to identify the  trace-preserving (TP) process. For the master equation form of evaluation, the algorithm to  identify the parameters with the time traces of an observable   was proposed in \cite{zhang2}. For the passive quantum system, Ref. \cite{7130587} discussed its identifiability and identification. 
 When the quantum process is non-trace-preserving (non-TP), a convex optimization method \cite{HUANG2020286} was proposed and 
  maximum likelihood estimation was utilized  \cite{PhysRevA.82.042307}  for non-TP process in experiment.
  For  non-Markovian dynamics, Refs. \cite{white2020demonstration} and \cite{PRXQuantum.3.020344} proposed  process tensor tomography.

According to the system architecture, there are 
generally three classes of QPT: Standard Quantum Process Tomography (SQPT) \cite{qci,unitary,PhysRevA.63.020101,xiaocdc2}, Ancilla-Assisted Process Tomography (AAPT) \cite{meaqo,aaqpt} and  Direct Characterization of Quantum Dynamics (DCQD) \cite{dcqd,dcgt,edcqy}. In SQPT, different states are inputted to the process and the output states are measured to collect information for QPT. In AAPT, an ancilla system is introduced besides the original system, and
 AAPT only needs to prepare a single quantum state with a full Schmidt number on the extended Hilbert space, while   SQPT needs at least $ d^2 $ different input states \cite{qci}. A widely used  AAPT input state is the maximally entangled state which has been generated in many experiments \cite{m2,m3,m4}.
DCQD relies on quantum error-detection techniques and besides  maximally entangled state, it also needs to generate different non-maximally entangled states \cite{dcgt}.

The main focus of this paper is on AAPT, and we adopt the two-stage solution (TSS) proposed in \cite{xiaocdc2} to address it. This method is an analytical algorithm, which enables us to analyze its computational complexity, error upper bound, and optimal design in greater detail.
Moreover, it can be applied for both TP and non-TP processes. In AAPT,
applying QST of the output state in the extended Hilbert space and TSS algorithm, we obtain an analytical solution of the reconstructed quantum process. Our algorithm has a computational complexity of $O(Md_A^2d_B^2)$ where $ M $ is the type number of the measurement operators, and $d_A$, $d_B$ are the dimensions of the quantum process of interest and the ancilla system, respectively. Then we establish an analytical error upper bound for our algorithm. Based on this error bound and the numerical stability, we further consider the optimization of the input state
 in AAPT.
We prove that the maximally entangled state is the optimal input state.
 A numerical example on a phase  damping process demonstrates the
effectiveness of the optimal design and illustrates the theoretical error analysis.

The organization of this paper is as follows. Section
II presents the two-stage solution for AAPT. Section III discusses the computational complexity, error analysis and the optimal design of the input state. A numerical example is presented in Section VI and Section V concludes this paper.

\section{Two-stage solution for ancilla-assisted quantum process tomography}
The quantum dynamics of a $ d $-dimensional system can be represented using a completely-positive (CP) linear map $ \mathcal{E} $.
We can construct a process matrix $ X \in \mathbb{C}^{d^2\times d^2} $ ($\mathbb{C}^{d^2\times d^2}$ is the set of all $d^2\times d^2$ complex matrices) from $ \mathcal{E} $ following the procedures in \cite{qci,8022944}, which is a Hermitian, positive semidefinite matrix. There is a one-to-one relationship between $ X $ and $ \mathcal{E} $, and the goal of QPT is to determine the process matrix $ X $.  We define the partial trace of $ X\in\mathbb{H}_A\otimes\mathbb{H}_B $  on Hilbert space $\mathbb{H}_A$ as $\text{Tr}_A(X)$.
When no information about the process output is lost, it holds $\operatorname{Tr}\left(\mathcal{E}\left(\rho^{\text{in}}\right)\right)=1$ where $ \rho^{\text{in}} $ is the input state and we call it a trace-preserving (TP) process. Otherwise, $\operatorname{Tr}\left(\mathcal{E}\left(\rho^{\text{in}}\right)\right)<1$ and the process is called  non-trace-preserving (non-TP) \cite{qci,8022944}.

A typical framework for QPT is ancilla-assisted quantum process tomography (AAPT).
In AAPT, an auxiliary system (ancilla) experiencing the identity channel is attached to the principal system, and the input state and the measurements on the output state are both on the extended Hilbert space as shown in Fig. \ref{AAPT}.
In this section, we firstly review AAPT's procedures  in \cite{aaqpt} and then we apply the two-stage solution (TSS) in \cite{xiaocdc2} to solve the AAPT problem for both TP and non-TP processes.
\begin{figure}[H]
	\centering
	\includegraphics[width=3.5in]{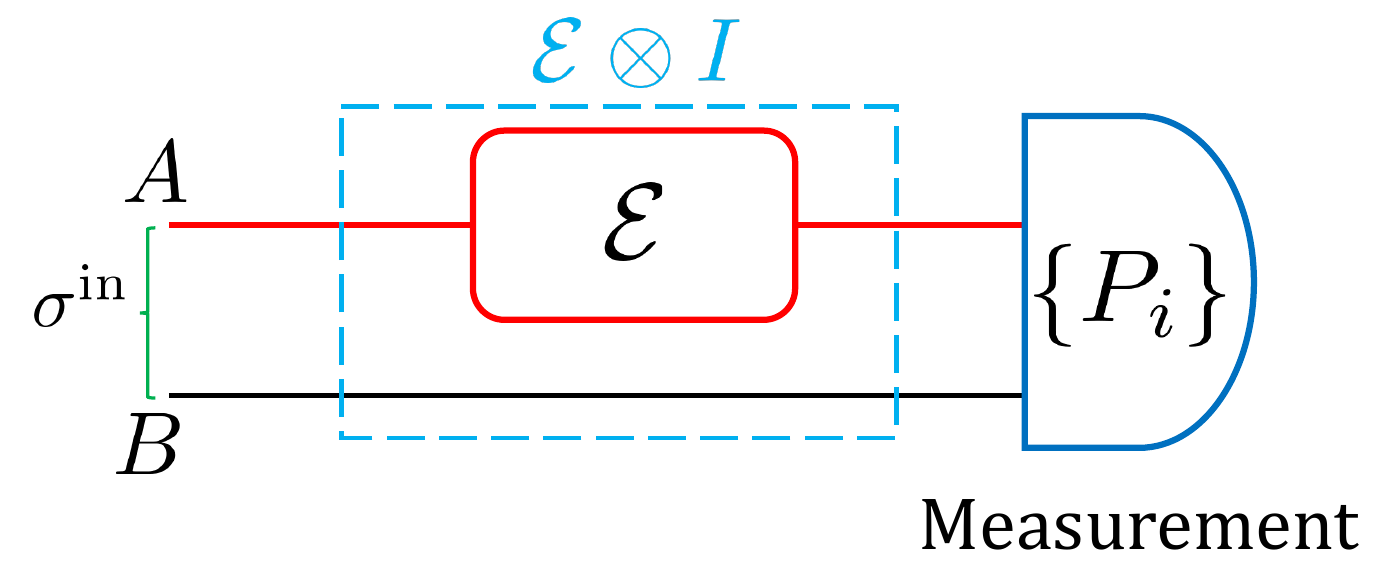}
	\centering{\caption{Schematic diagram of AAPT. The measurement operators $ \{P_i\} $ can be separable or entangled operators.}\label{AAPT}}
\end{figure}
For the quantum system $ A$ and the process $\mathcal{E}$, AAPT utilizes an ancilla system $B$ as Fig. \ref{AAPT} whose Hilbert space dimension is not smaller than that of $A$, i.e., $ d_B \geq d_A $ \cite{aaqpt}.
For the input state $ \sigma^{\text{in}} $, its operator--Schmidt decomposition \cite{PhysRevA.67.052301} is
\begin{equation}\label{instate}
\sigma^{\text{in}}=\sum_{i=1}^{d_{A}^{2}} s_i A_i \otimes B_i,
\end{equation}
where $ \otimes $ is the tensor product and the $s_l$ are non-negative real numbers. The sets $\left\{A_i\right\}$ and $\left\{B_i\right\}$ form orthonormal operator bases for systems $A$ and $B$, respectively with the inner product  $\langle X, Y\rangle\triangleq\text{Tr}(X^\dagger Y)$ \cite{aaqpt}. The Schmidt number of the input state $\operatorname{Sch}(\sigma^{\text{in}})$ is defined as the number of nonzero terms $ s_l $ in the Schmidt decomposition. AAPT requires $\operatorname{Sch}(\sigma^{\text{in}})=d_{A}^2$, i.e., $ s_i>0, \forall 1\leq i \leq d_A^2 $ where $ \sigma^{\text{in}} $ can be separable or entangled \cite{aaqpt}. Since  the set of states with the Schmidt number less than $d_A^2$ is of zero measure, almost all the states  of the combined system $A B$ can be used for AAPT \cite{qptre}.

After the process $ \mathcal{E} $, the output state is
\begin{equation}
\sigma^{\text{out}}=(\mathcal{E} \otimes I)(\sigma^{\text{in}})=\sum_{i=1}^{d_{A}^{2}} s_i \mathcal{E}\left(A_i\right) \otimes B_i.
\end{equation}
Let   the conjugate $ (*) $ and transpose $ (T) $ of $ B $ be $B^\dagger  $.
Since
\begin{equation}\label{ptraceout}
\begin{aligned}
&\operatorname{Tr}_B\left[\left(I_{d_A} \otimes B_j^{\dagger}\right) \sigma^{\text{out}}\right]\\
=&\sum_{i=1}^{d_{A}^{2}} s_i \mathcal{E}\left(A_i\right) \operatorname{Tr}\left(B_j^{\dagger} B_i\right)=s_j \mathcal{E}\left(A_j\right),
\end{aligned}
\end{equation} 
we have
\begin{equation}\label{ea}
\mathcal{E}\left(A_j\right)=\operatorname{Tr}_B\left[\left(I_{d_A} \otimes B_j^{\dagger}\right) \sigma^{\text{out}}\right] /s_j.
\end{equation}
Therefore, we can obtain $ d_{A}^2 $ linearly independent input-output relationships.
Define the  parameterization matrix of $ \{A_j\}_{j=1}^{d_A^2} $ and $\left\{{\mathcal{E}}\left(A_j\right)\right\}_{j=1}^{d_A^2}$ as
\begin{equation}\label{yv}
\begin{aligned}
V &\triangleq\left[\operatorname{vec}\left(A_1\right), \operatorname{vec}\left(A_2\right), \ldots, \operatorname{vec}\left(A_{d_A^2}\right)\right],\\
 Y &\triangleq\left[\operatorname{vec}\left({\mathcal{E}}\left(A_1\right)\right), \operatorname{vec}\left({\mathcal{E}}\left(A_2\right)\right), \ldots, \operatorname{vec}\left({\mathcal{E}}\left(A_{d_A^2}\right)\right)\right],
\end{aligned}
\end{equation}
which are both $ d_{A}^2 \times d_{A}^2 $  matrices and the vectorization function is defined as 
\begin{equation}
\begin{aligned}
\text{vec}(H_{m\times n})\triangleq&[(H)_{11},(H)_{21},\cdots,(H)_{m1},(H)_{12},\\&\cdots,(H)_{m2},
\cdots,(H)_{1n},\cdots,(H)_{mn}]^T.
\end{aligned}
\end{equation}
We also define that $\text{vec}^{-1}(\cdot)$ maps a $d^2\times 1$ vector into a $d\times d$ square matrix.
The relationship between $ V $, $ Y $ and process matrix $ X $  is
\begin{equation}\label{m1}
	\left(I_{d^2} \otimes V^T\right) R\operatorname{vec}(X)=\operatorname{vec}(Y),
\end{equation}
where $ R $ is a determined permutation matrix.
Since $ \{A_j\}_{j=1}^{d_A^2} $ forms an orthonormal operator base,  we have $ V $ is a unitary matrix and thus $\left(V^T\right)^{-1}=V^{*}$.

In experiment, we reconstruct the output state $ \hat\sigma^{\text{out}} $ by QST. Assuming that the number of copies for the output state is $ N $ and the measurement operators are $ \{P_m\}_{m=1}^{M} $ which are applied on the output state. Let the total set number of the measurement operators (i.e., measurement basis set) be $ L $ and the type number of the measurement operators be $ M $. Here we assume that these measurement operators are informationally complete and thus $ M\geq d_A^2d_B^2 $. 
 Using \eqref{ea}, we calculate $\left\{\hat{\mathcal{E}}\left(A_m\right)\right\}_{m=1}^{d_A^2}$  and construct $ \hat Y $ from  $ \hat\sigma^{\text{out}} $.
Since $I_{d_A^2} \otimes V^T  $ is a unitary matrix, the solution of $\operatorname{vec}\left(\hat{X}\right)$ is 
 \begin{equation}
 \operatorname{vec}(\hat X)=R^{T}\left(I_{d_A^2} \otimes V^{*}\right) \operatorname{vec}(\hat{Y}).
 \end{equation}
 Then we consider the the positive semidefinite requirement $X \geq 0$ and the constraint $\operatorname{Tr}_{A}(X)= I_{d_{A}}$ for the TP process and $\operatorname{Tr}_{A}(X)\leq I_{d_{A}}$ for the non-TP process \cite{8022944,PhysRevA.82.042307}. Thus, the QPT can be 
converted into the following optimization problem: Given the parameterization matrix $V$ of $ \{A_m\}_{m=1}^{d_A^2} $, the determined permutation matrix $ R $ and reconstructed $\hat{Y}$, find a Hermitian and positive semidefinite estimate $\hat{X}$ minimizing
	$$\left\|\hat X-\operatorname{vec}^{-1}\left(R^{T}\left(I_{d_{A}^2} \otimes V^{*}\right) \operatorname{vec}(\hat{Y})\right)\right\|,$$ 
	such that $\operatorname{Tr}_{A}(\hat X)= I_{d_{A}}$ for the TP process or  $\operatorname{Tr}_{A}(\hat X)\leq I_{d_{A}}$ for the non-TP process.

 In this paper, $||X||$ denotes the Frobenius norm of $X$.
Here we apply the TSS algorithm in \cite{xiaocdc2} to solve this problem  because it is an analytical algorithm and thus we can further analyze the computational complexity and error bound.   Let $$\hat G\triangleq\operatorname{vec}^{-1}\left(R^{T}\left(I_{d_A^2}\otimes  V^{*}\right) \operatorname{vec}(\hat{Y})\right)$$ be a given matrix. Firstly, we aim to find a Hermitian and positive semidefinite $d_{A}^2\times d_{A}^2$ matrix $\hat D$ minimizing $||\hat D-\hat G||$. We perform the spectral decomposition as $\frac{\hat G+\hat G^{\dagger}}{2}=W\hat{K}W^{\dagger}  $  where $\hat{K}=\operatorname{diag}\left\{k_{1}, \cdots, k_{d_{A}^{2}}\right\}$ is a  diagonal matrix.
Let
\begin{equation}
z_{i}= \begin{cases}
k_{i}, & k_{i} \geq 0, \\ 
0, & k_{i}<0,
\end{cases}
\end{equation}
and the unique optimal solution is $ \hat D=W\operatorname{diag}\{z\}W^{\dagger}  $.

Then we define $ \hat{E}\triangleq\operatorname{Tr}_{A}(\hat{D}) $. For a TP process, we can assume that  $\hat E >0  $ because $ \hat E$  converges to $ I_{d_A} $ as $ N $ tends to infinity. Thus the final estimate is 
\begin{equation}\label{tss21}
\hat{X}=(I_{d_A} \otimes \hat{E}^{-1/2}) \hat{D}(I_{d_A} \otimes \hat{E}^{-1/2})^{\dagger},
\end{equation}
which satisfies $ \hat{X}\geq 0 $ and $\operatorname{Tr}_{A}(\hat X)= I_{d_{A}}$.
For a non-TP process, assuming that the spectral decomposition of $ \hat E $ is
\begin{equation}\label{hatf}
\hat{E}=\hat U \operatorname{diag}\left\{\hat e_{1}, \cdots, \hat e_{d_A}\right\}\hat U^{\dagger},
\end{equation}
where $ \hat e_{1} \geq\cdots \geq\hat e_{c}>0$ and $ \hat e_{c+1}= \cdots =\hat e_{d_A}=0 $.
Then we define
\begin{equation}\label{barf}
\bar{E}\triangleq \hat U \operatorname{diag}\left\{\bar e_{1}, \cdots, \bar e_{d_A}\right\}\hat U^{\dagger},
\end{equation}
where $ \bar e_{i}=\hat e_{i} $ for $ 1\leq i \leq c $, $ \bar e_{i}={\hat e_{c}}/{N} $ for $ c+1\leq i \leq d_A$, and $ N $ is the copy number. 
Since $ \bar{E} $ is invertible, we also define
\begin{equation}
\tilde{E}\triangleq\hat U \operatorname{diag}\left\{\tilde{e}_{1}, \cdots, \tilde{e}_{d_A}\right\}\hat U^{\dagger},
\end{equation}
where $ \tilde{e}_{i}=\min\left({\bar e}_{i},1\right) $ for $ 1\leq i\leq d_A $.
The final estimate is 
\begin{equation}\label{tss22}
\hat{X}=(I_{d_A} \otimes  \tilde  E^{1/2}\bar E^{-1/2}) \hat{D}(I_{d_A} \otimes  \tilde  E^{1/2}\bar E^{-1/2})^{\dagger},
\end{equation}
which satisfies the constraints $\hat X \geq 0$ and  $\operatorname{Tr}_{A}(\hat X) \leq I_{d_{A}}$ for the non-TP process.
In summary, Fig. \ref{aapt2} is the total procedures of the AAPT  with TSS algorithm.
\begin{remark}
	In fact, when  $ d_A>d_B $, we can also complete the task of AAPT by preparing different input states. Assume that there are $ F$ types of input states $ \{\sigma_i^{\text{in}}\}_{i=1}^{F} $ and the Schmidt decomposition is
	\begin{equation}
	\sigma_i^{\text{in}}=\sum_{l=1}^{d_B^2} s_l A_l^i \otimes B_l^i.
	\end{equation}
	To obtain a unique estimate of the process $ X $, we need to ensure the number of linearly independent operators in $ \left\{A_l^i\right\}_{l=1,i=1}^{d_{B}^{2},F} $ is larger or equal to $ d_A^2 $.
	Thus, if $ d_B < d_A $, we need at least $\lceil{d_A^2}/{d_B^2}\rceil $ different  input states with $\operatorname{Sch}(\sigma_{i}^{\text{in}})=d_B^2$, where the least integer greater than or equal to $x$ is defined as $ \lceil x \rceil $.	
	A special case is $ d_B=1 $ which is in fact SQPT and we need at least $ d_A^2 $ different input states.
\end{remark}

\section{Computational complexity and error analysis}
Since the TSS algorithm is analytical, in this section, we analyze the computational complexity and error upper bound. 
\subsection{Computational complexity}
\begin{figure}
	\centering
	\includegraphics[width=3.5in]{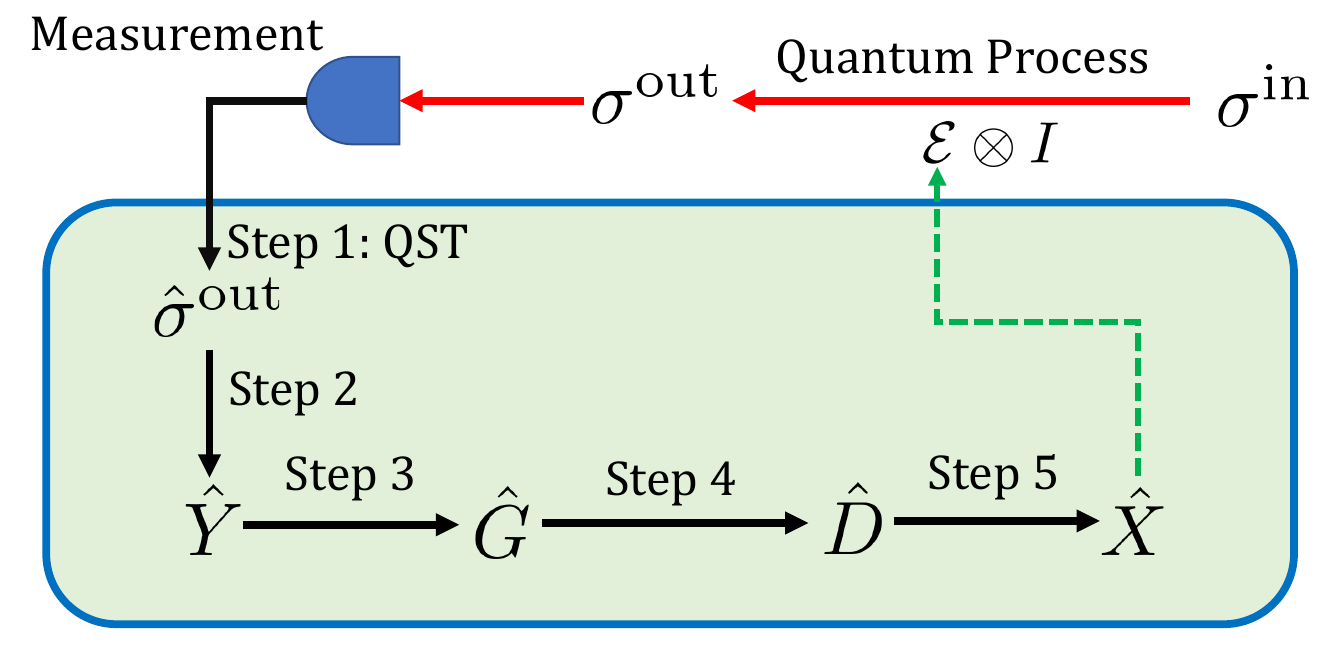}
	\centering{\caption{Procedures of AAPT with the TSS algorithm.}\label{aapt2}}
\end{figure}

In this paper, we neglect the time required for conducting experiments and focus solely on the computational complexity associated with each step outlined in the blue box of Figure \ref{aapt2}.

\textbf{Step 1}. In this step, we apply the quantum state tomography on the output state $ \hat\sigma^{\text{out}} $ using the LRE method in  \cite{Qi2013}, where
the computational complexity  is  $O(Md_A^2d_B^2)$.

\textbf{Step 2}. In this step, we use \eqref{ea} to construct $\hat{\mathcal{E}}\left(A_m\right)$ and $\hat Y$.  The computational complexity is determined by \eqref{ea} which is not worse than $ O(d_A^5d_B^3)  $.

\textbf{Step 3}. In this step, we calculate $\hat G\triangleq\operatorname{vec}^{-1}(R^{T}(I_{d_{A}^2}\otimes  V^{*}) \operatorname{vec}(\hat{Y}))$. The computational complexity for $ (I_{d_{A}^2}\otimes  V^{*}) \operatorname{vec}(\hat{Y})$ 
is $ O(d_A^6) $. Then the computational complexity for $ R^{T}\left(I_{d_A^2}\otimes  V^{*}\right) \operatorname{vec}(\hat{Y}) $ 
is $ O(d_A^4) $ because $ R $ is a permutation matrix.

\textbf{Steps 4--5}. In these steps, we apply the TSS algorithm in \cite{xiaocdc2} and the computational complexity is $O(d_A^6)$ \cite{xiaocdc2}.

Hence, the overall computational complexity of the AAPT with the TSS algorithm can be expressed as $O(Md_A^2d_B^2)$ which is dominated by the QST on the output state $ \hat\sigma^{\text{out}} $.

\subsection{Error analysis}
Here we present the following theorem to describe the analytical error  upper bound  for the AAPT with TSS algorithm.
\begin{theorem}\label{t1}
	Using the TSS algorithm in AAPT, the estimation error  $ \mathbb{E}\|\hat{X}-X\| $  scales as
	$$ O\!\!\left(\frac{d_{A}d_{B}^{1/2}\operatorname{Tr}(E)\sqrt{L\operatorname{Tr}\left(\left({C}^{\dagger} {C}\right)^{-1}\right)}}{\sqrt{N}}  \sqrt{\sum_{j=1}^{d_{A}^2} \frac{1}{s_j^2}}\right), $$
	where $N$ is the number of copies for the output state, $ L $ is the number of POVM sets, $ {C} $ is the parameterization matrix for the measurement operators, $ E=\operatorname{Tr}_{A} (X) $ and $\mathbb E(\cdot)$ represents the expectation taken over all possible measurement outcomes.
\end{theorem}
\begin{proof}
	\subsubsection{Error in Step 1} Step-1 is the QST on the output state $ \hat\sigma^{\text{out}} $ using LRE.
	From \cite{Qi2013}, we have
	\begin{equation}
	\mathbb{E}\left\|\hat{\sigma}^{\text {out }}-\sigma^{\text {out }}\right\|^2 \leq \frac{L}{4 N} \operatorname{Tr}\left(\left({C}^{\dagger} {C}\right)^{-1}\right).
	\end{equation}
\subsubsection{Error in Step 2}	
	Then in Step 2, we reconstruct $\hat{\mathcal{E}}\left(A_j\right)$ and $ \hat Y $ from $ \hat{\sigma}^{\text {out }} $.
	We introduce the following Lemma:
	\begin{lemma}\cite{dbound}\label{lemma2}
		Let $\mathbb{H}_A$ and $\mathbb{H}_B$ be finite-dimensional Hilbert spaces of dimensions $d_{A}$ and $d_{B}$, respectively, and let $X \in \mathbb{H}_A \otimes \mathbb{H}_B$. Then for any unitarily invariant norm that is multiplicative over tensor products, the partial trace satisfies the norm inequality
		\begin{equation}
		\left\|\operatorname{Tr}_{A}(X) \right\| \leq \frac{d_{A}}{\left\|I_{A}\right\|}\|X\|,
		\end{equation}
		where $I_{A}$ is the identity operator.
	\end{lemma}
	
	 Using Lemma \ref{lemma2} and  $ \left\|I_{d_A} \otimes B_j^{\dagger}\right\|=\sqrt{d}\left\|B_j\right\| $, we have
	\begin{equation}
	\begin{aligned}
	&\left\|\hat{\mathcal{E}}\left(A_j\right)-{\mathcal{E}}\left(A_j\right)\right\| \\
	=&\|\operatorname{Tr}_B\left[\left(I_{d_A} \otimes B_j^{\dagger}\right) \hat{\sigma}^{\text {out }}\right] / s_j\\
	&-\operatorname{Tr}_B\left[\left(I_{d_A} \otimes B_j^{\dagger}\right) \sigma^{\text {out }}\right] / s_j\| \\
	\leq& \frac{\sqrt{d_B}}{s_j}\left\|I_{d_A} \otimes B_j^{\dagger}\right\|\left\|\hat{\sigma}^{\text {out }}-\sigma^{\text {out }}\right\| \\
	=&\frac{\sqrt{d_Ad_B}}{s_j}\left\|B_j\right\|\left\|\hat{\sigma}^{\text {out }}-\sigma^{\text {out }}\right\|.
	\end{aligned}
	\end{equation}
	Since $ \left\|B_j\right\|^2=1 $, the corresponding error is bounded by
	\begin{equation}\label{aqy}
	\begin{aligned}
	&\mathbb{E}\|\hat{Y}-Y\|^2=\mathbb{E}\|\operatorname{vec}(\hat{Y})-\operatorname{vec}(Y)\|^2\\
	=&\mathbb{E}\sum_{j=1}^{d_{A}^2}\left\|\hat{\mathcal{E}}\left(A_j\right)-{\mathcal{E}}\left(A_j\right)\right\|^2 \\
	\leq& \frac{d_{A} d_BJ}{4 N} \operatorname{Tr}\left(\left({C}^{\dagger} {C}\right)^{-1}\right) \sum_{j=1}^{d_{A}^2} \frac{1}{s_j^2}.
	\end{aligned}
	\end{equation}
\subsubsection{Error in Step 3}	In Step 3, since $ V^{*} $ is a unitary matrix, 	we have
\begin{equation}\label{yy}
\begin{aligned}
&\quad\|\hat{G}-G\|^{2}\\
&=\left\|\operatorname{vec}\left(\hat G\right)-\operatorname{vec}\left( G\right)\right\|^2 \\
&=\left\|R^{T}\left(I_{d_A^2} \otimes V^{*}\right)(\operatorname{vec}(\hat{Y})-\operatorname{vec}(Y))\right\|^{2} \\
&=\|\operatorname{vec}(\hat{Y})-\operatorname{vec}(Y)\|^{2}. 
\end{aligned}
\end{equation}
\subsubsection{Error in Steps 4--5}		
	In Steps 4--5, we apply the TSS algorithm in \cite{xiaocdc2}. For TP processes, Ref. \cite{xiaocdc2} has proved that 
	\begin{equation}\label{ggerror}
	\begin{aligned}
	\|\hat{D}-D\| \leq\|\hat{D}-\hat{G}\|+\|\hat{G}-G\| \leq 2\|\hat{G}-G\|,
	\end{aligned}
	\end{equation}
	and
	\begin{equation}\label{xg1}
	\begin{aligned}
	\|\hat{X}-X\| &\leq\|\hat{X}-\hat{D}\|+\|\hat{D}-D\|\\
	& \leq (d_A^{3/2}+1)\|\hat{D}-D\|. \\
	\end{aligned}
	\end{equation}
	For non-TP processes, using the similar error analysis in \cite{xiaocdc2}, we have
		\begin{equation}\label{xg}
	\begin{aligned}
	\|\hat{X}-X\| &\leq\|\hat{X}-\hat{D}\|+\|\hat{D}-D\|\\
	& \leq (\sqrt{d_A}\operatorname{Tr}({E})+1)\|\hat{D}-D\|, \\
	\end{aligned}
	\end{equation}
	where $ E=\operatorname{Tr}_{A}\left(X\right) $.
	Using \eqref{aqy}, \eqref{yy}, \eqref{ggerror}, and \eqref{xg}, the final error upper bound is
	\begin{equation}\label{finalerror}
	\begin{aligned}
	&\mathbb{E}\|\hat{X}-X\| \leq 2\left(\sqrt{d_A}\operatorname{Tr}(F)+1\right)\mathbb{E}\|\hat{D}-D\| \\
	\sim&O\left(\frac{d_{A}d_{B}^{1/2}\operatorname{Tr}(E)\sqrt{L\operatorname{Tr}\left(\left({C}^{\dagger} {C}\right)^{-1}\right)}}{\sqrt{N}}  \sqrt{\sum_{j=1}^{d_{A}^2} \frac{1}{s_j^2}}\right).
	\end{aligned}
	\end{equation}
\end{proof}

\subsection{Optimal input state}
In this section, based on the error bound \eqref{finalerror} and the numerical stability, we consider the optimization of the input state $ \sigma^{\text{in}} $.
 Based on the numerical stability from \eqref{ea}, we aim to maximize the minimum of $ \{s_j\}_{j=1}^{d_{A}^{2}} $.
Refs. \cite{meaqo} and \cite{qptre} also considered this problem and gave a definition of faithfulness. They have shown that the maximally entangled state is the optimal faithful (in the sense of minimal experimental errors) input state.
Besides \eqref{ea}, we also aim to minimize $ \sum_{j=1}^{d_{A}^2} \frac{1}{s_j^2} $ based on the error bound \eqref{finalerror}. Since
\begin{equation}
\operatorname{Tr}\left(\left(\sigma^{\text{in}}\right)^2\right)=\sum_{j=1}^{d_{A}^2} s_j^2 \leq 1,
\end{equation}
using Cauchy–Schwarz inequality, we have
\begin{equation}\label{opst1}
\sum_{j=1}^{d_{A}^2} \frac{1}{s_j^2} \sum_{j=1}^{d_{A}^2} s_j^2 \geq d_{A}^4, \sum_{j=1}^{d_{A}^2} \frac{1}{s_j^2} \geq d_{A}^4,
\end{equation}
and
\begin{equation}\label{opst2}
\min \left(\{s_j\}_{j=1}^{d_{A}^{2}}\right)\leq \frac{1}{d_A}.
\end{equation}
The inequalities of \eqref{opst1} and \eqref{opst2} becomes equalities \emph{simultaneously} if and only if $ \sigma^{\text{in}} $ is a pure state and  $ s_j=\frac{1}{d_A} $, i.e., the input state $ \sigma^{\text{in}} $ is the maximally entangled state. Therefore, the optimal input state for AAPT is the maximally entangled state.

\section{Numerical examples}
In practice, different measurement bases can be employed to perform measurement. In this section, we use Cube measurements \cite{PhysRevA.78.052122} for AAPT because it is relatively easy to be realized in experiment. For one-qubit systems, the Cube measurements are $\left\{\frac{I \pm \sigma_{x}}{2}, \frac{I \pm \sigma_{y}}{2}, \frac{I \pm \sigma_{z}}{2}\right\}$ where $ \sigma_{x}=\left(\begin{array}{ll}
0 & 1 \\
1 & 0
\end{array}\right) $, $ \sigma_{y}=\left(\begin{array}{cc}
0 & -\mathrm i \\
\mathrm i & 0
\end{array}\right) $, $\sigma_{z}=\left(\begin{array}{cc}
1 & 0 \\
0 & -1
\end{array}\right)$ are Pauli matrices. For  multi-qubit systems, the Cube measurements are the tensor products of one-qubit Cube measurements.

 We consider a phase  damping process (\cite{qci}, Page 384) for one-qubit system which is TP and can be characterized by two Kraus operators
\begin{equation}\label{phase}
\mathcal{A}_1=\left(\begin{array}{cc}
1 & 0 \\
0 & \sqrt{1-\lambda}
\end{array}\right),\mathcal{A}_2=\left(\begin{array}{cc}
0 & 0 \\
0 & \sqrt{\lambda}
\end{array}\right),
\end{equation}
where $\lambda=2/3 $ describes the probability that a photon from
the system has been scattered without loss of energy \cite{qci}.

For AAPT, we input the maximally entangled state which is optimal input state in this paper, and apply two-qubit Cube measurements on the output state. Then using the algorithm in \cite{qetlab,MISZCZAK2012118}, we generate one random two-qubit input state with the full Schmidt number  for AAPT.
Thus, the total resource number in AAPT is $ N $ which is the copy number for the unique input state in the extended Hilbert space. For each resource number, we repeat our algorithm $ 100 $ times, and obtain the average MSE (mean squared error) and error bars.
The MSE with the optimal and random input states versus  the total resource number $N$ in logarithm are shown in Fig. \ref{mse} where the scalings of the MSEs are both $ O(1/N) $ satisfying Theorem \ref{t1}. In addition, the MSE with the optimal input state are smaller than the MSE with random input states, which demonstrates the performance of the optimal design for the input state.

\begin{figure}
	\centering
	\includegraphics[width=3.5in]{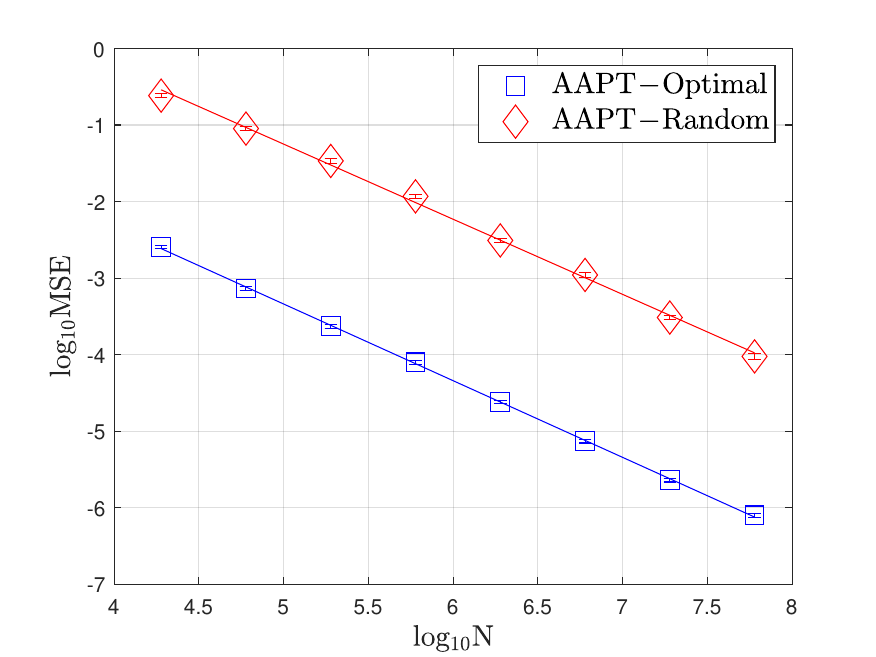}
	\centering{\caption{MSE versus the total resource number $N_t$ with the optimal input state (the maximally entangled state) and a random input state.}\label{mse}}
\end{figure}

\section{Conclusion}
In this paper, we have applied the  two-stage solution (TSS) for AAPT. We have analyzed the computational complexity and established 
an error upper bound. Furthermore, we have discussed the optimal design on the input state in AAPT. 
Numerical example on a phase  damping process has demonstrated the
effectiveness of the optimal design and illustrated the theoretical error analysis. Further work will focus on the optimal design about the measurement operators.



\bibliographystyle{ieeetr}         
\bibliography{cdcaqpt} 

\end{document}